\documentclass[final,authoryear,1p,times]{elsarticle}

\usepackage{}
\usepackage{amsbsy}
\usepackage{mathrsfs}
\usepackage[fleqn]{amsmath}
\usepackage{amssymb}

\usepackage{amsfonts}
\usepackage{amssymb}
\usepackage{multicol}
\usepackage{color}
\usepackage{caption}
\usepackage{graphicx,float}
\usepackage{multirow}
\usepackage{color}
\usepackage{array}
\usepackage{longtable}
\usepackage{calc}
\usepackage{multirow}
\usepackage{hhline}
\usepackage{ifthen}
\usepackage{amsmath}
\usepackage{amsfonts}
\usepackage{amssymb}
\usepackage{amsthm}
\usepackage{subfig}
\usepackage{booktabs}

\newtheorem{lemma}{Lemma}
\newtheorem{remark}{Remark}
\newtheorem{assumption}{Assumption}

\usepackage{hyperref}

\begin{document}

\begin{frontmatter}


\title{Multivariate Self-Exciting Threshold Autoregressive Models with eXogenous Input \\ \{PRELIMINARY VERSION-- please do not quote\}}
\author{Peter Martey ADDO\corref{cor1}}

\address{European Doctorate in Economics--Erasmus Mundus (EDEEM)\\ Centre d'\'{E}conomie de la Sorbonne (CES) - CNRS : UMR8174 - Universit\'{e} Paris I - Panth\'{e}on Sorbonne\\
Universit\`{a} Ca'Foscari of Venice, Department of Economics\\ email: peter.addo@univ-paris1.fr}


\begin{abstract}

\noindent 
This study defines a multivariate Self--Exciting Threshold Autoregressive with eXogenous input (MSETARX) models and present an estimation procedure for the parameters. The conditions for stationarity of the nonlinear MSETARX models is provided. In particular, the efficiency of an adaptive parameter estimation algorithm and LSE (least squares estimate) algorithm for this class of models is then provided via simulations. 

\end{abstract}

\end{frontmatter}

\section{Introduction}\label{sec:intro}

Recently there has been considerable interest in nonlinear time series analysis (\cite{Priestley88, Tong90, Brock91, Granger93, Terasvirta94, Hansen11, pma5}, and references therein), due primarily to the various limitations encountered with linear time series models in real applications. Many nonlinear time series models have been introduced in the literature and illustrated to be useful in some 
applications (\cite{Granger78, Priestley88, Subba84, Haggan81, Tong83, Tong90}). For instance, \cite{Tong78,Tong90} proposed the threshold autoregressive (TAR) model and showed its usefulness in describing the asymmetric limit cycle of the annual sunspot number. 
Let $(\Omega, {\cal F}, P)$ be a probability space, $R=\bigcup_{j=1}^{l}R_{j}$, $R_{j}=(r_{j-1},r_{j}], -\infty=r_{0}<r_{1}<\cdots < r_{l}=\infty$ a disjunctive decomposition of the real axis. Let $d,p_{1}, \cdots, p_{l} \in Z^{+}$. Any solution of $(y_{t})_{t}$ of 

\begin{equation}\label{setar}
 y_{t}+\sum_{j=1}^{l}y_{t,d}^{(j)}\left(a_{0}^{(j)}+\sum_{i=1}^{p_{l}}a_{i}^{(j)}y_{t-i}\right)=\sum_{j=1}^{l}y_{t,d}^{(j)}\varepsilon_{t}^{(j)}
\end{equation}
 where 
\begin{equation}
 y_{t,d}^{(j)}=
\begin{cases}
&1 ; \quad y_{t-d}\in R_{j} \quad \\ 
&0 ; \quad y_{t-d} \notin R_{j}.\\
\end{cases}
\end{equation} is a univariate Self--Exciting Threshold Autoregressive process denoted by SETAR $(l,p_{1},\cdots, p_{l})$ with delay $d$ (see \cite{Tong83,Tong90} and the references therein). The process $(y_{t})_{t}$ is assumed to be ergodic and its stationary distribution has a finite second moment. The process $(\varepsilon_{t})_{t}^{(j)}$ in model equation \eqref{setar} for each regime $j$ is assumed to be a martingale difference sequence with respect to an increasing sequence of $\sigma$-field, denoted as $\mathcal{F}_{t}$, i.e., $E[\varepsilon_{t}^{(j)}|\mathcal{F}_{t-1}]=0$. 
In this setting, the conditional variance of the process $(\varepsilon_{t})_{t}^{(j)}$ can be a  constant, $E[(\varepsilon_{t}^{(j)})^{2}|\mathcal{F}_{t-1}]=\sigma^{2}$ or allowed for possibly asymmetric autoregressive conditional heteroscedasticity. The model equation \eqref{setar} is nonlinear in time when the number of regimes $l>1$ and is a piecewise linear model in the threshold space $y_{t-d}$. 
Thus SETAR model \eqref{setar} adopts a piecewise linear setting in such a fashion that regime switches are triggered by an observed variable crossing an unknown threshold. For a review on the asymptotic theory and inference for the SETAR model \eqref{setar}, see \cite{Tong90, Chan93, Qian98, Hansen97, Hansen99, Hansen00}. Despite the simplicity of SETAR models, they have been shown to be able to capture economically interesting asymmetries, regime changes (such as periods of low/high stock market valuations, recessions/expansions, periods of low/high interest rates, etc), and empirically observed nonlinear dynamics relevant to economic data. For instance, \cite{Pfann96} used a single--threshold SETAR model in describing the dynamic behaviour of the three--month US T-bill interest rate. 

In analysing multivariate relationships between economic variables, the linear Vector Autoregression (VAR) models have gain popularity for empirical 
macroeconomic modelling, policy analysis and forecasting. However, the inability of these linear models to capture non-linear dynamics such as regime switching and asymmetric responses to shocks, has gained attention in macroeconomic research. For example, a significant number of empirical studies document asymmetries in the effects of monetary policy on output growth (\cite{Philip99} and reference therein). In this respect, the interest in nonlinear ARX time series and regression models has been increasing in econometrics as in other disciplines (\cite{Granger93, ChTs93, Terasvirta13} and references therein). In this work, we consider the introduction of an exogenous input $(\mathbf{f}_{t})_{t}$ as an extension of the Multivariate SETAR model formulation and has a structural form of a nonlinear bivariate ARX model (\cite{EMDT97}). Unlike the multivariate threshold model proposed in \cite{Tsy98}, we allow the possibility of the threshold variable to also be a multivariate process. In this case, the regime of the whole system is not necessarily determined by a single stationary subprocess. In otherwords, there exists thresholds for all subprocess of the multivariate process.

A short overview of the paper is as follows. In Section \ref{sec:MSETARX} we define the multivariate SETAR process with exogenous input denoted MSETARX model as an extension of the multivariate SETAR model. In Section \ref{sec:stationarity} we find conditions for stationarity of the MSETARX models, whereas Section \ref{sec:estimation} is used to present the LSE (least squares estimate) algorithm and an adaptive parameter estimation algorithm (\cite{Arnold01,Lutz06}) based on the stochastic gradient principles for linear systems shown to be suitable for nonlinear systems. The performance of the proposed algorithms for estimating the parameters of Multivariate SETARX models is evaluated via simulations in Section \ref{sec:simulations}. In Section \ref{sec:conclusion}, the modeling procedure for the MSETARX models and problems of estimation are briefly considered.

\section{Multivariate SETARX models}\label{sec:MSETARX}

Consider a $D$-dimensional time series $\mathbf{y}_{t}=(y_{1t},\cdots,y_{Dt})^{T}$ such that $L_{1},\cdots,L_{D} \in Z^{+}$, for each $1\leq i\leq D$, $(R_{j}^{i})_{j=1,2,\cdots,L_{i}}$ a disjuction decomposition of the real axis: 
$R=\bigcup_{j=1}^{L_{i}}R_{j}^{i}$ ; $i \in \{1,\cdots,D\}$. Let $L=\max \{L_{1},L_{2},\cdots,L_{D}\}$ and $R_{j}^{i}=\Phi$ ; $j=L_{i}+1, \cdots, L$. Then any solution $(\mathbf{y}_{t})_{t}$ of 
\begin{equation}\label{msetar}
 \mathbf{y}_{t}+\sum_{J\in \{1,\cdots, L\}^{D}}y_{t,d}^{(J)}\left(a_{0}^{(J)}+\sum_{i=1}^{p_{J}}A_{i}^{(J)}\mathbf{y}_{t-i}\right)=\sum_{J\in \{1,\cdots, L\}^{D}}y_{t,d}^{(J)}\mathbf{\varepsilon}_{t}^{(J)}
\end{equation} is called a multivariate SETAR process denoted MSETAR $(L,p_{J};J\in \{1,\cdots, L\}^{D})$, where $y_{t, d}^{(J)}: \{1,\cdots, L\}^{D} \longleftrightarrow \{0,1\}$ is the indicator variable defined by the following relation: 
$$\left(y_{t}^{(j_{1},\cdots,j_{D})} = 1 \right) \Leftrightarrow_{def} \left( (\mathbf{y}_{t-d})_{i} \in R_{j}^{i} ; j\in (1, \cdots, L)^D ; i \in (1, \cdots, D) \right)$$ and $\{\mathbf{\varepsilon}_{t}^{(J)}, \mathcal{F}_{t}\}$ be a sequence of martingale difference  with respect to an increasing sequence of $\sigma$-field $\{\mathcal{F}_{t}\}$ such that  
 $$\sup_{t\geq 0} E [\lVert \mathbf{\varepsilon}_{t+1}^{(J)} \rVert | \mathcal{F}_{t}]=0 \quad a.s, \quad \sup_{t\geq 0} E [\lVert \mathbf{\varepsilon}_{t+1}^{(J)} \rVert^{2} | \mathcal{F}_{t}]=\sigma^{2} < \infty \quad a.s,\quad \sup_{t\geq 0} E [\lVert \mathbf{\varepsilon}_{t+1}^{(J)} \rVert^{\alpha} | \mathcal{F}_{t}] < +\infty \quad a.s \quad $$ for 
 some $\alpha > 2$ and $\lVert \cdot \rVert$ be a matrix norm. 
 
Now consider a $D$-dimensional time series $\mathbf{y}_{t}=(y_{1t},\cdots,y_{Dt})^{T}$ and a $\kappa$-dimensional inputs $\mathbf{f}_{t}=(f_{1t},\cdots,f_{\kappa t})^{T}$  such that $L_{1},\cdots,L_{D} \in Z^{+}$, for each $1\leq i\leq D$, $(R_{j}^{i})_{j=1,2,\cdots,L_{i}}$ a disjuction decomposition of the real axis: 
$R=\bigcup_{j=1}^{L_{i}}R_{j}^{i}$ ; $i \in \{1,\cdots,D\}$. Let $L=\max \{L_{1},L_{2},\cdots,L_{D}\}$ be the maximum of the number of regimes for 
each subprocess of $\mathbf{y}_{t}$ and $R_{j}^{i}=\Phi$ ; $j=L_{i}+1, \cdots, L$. Then any solution $(\mathbf{y}_{t})_{t}$ of 
\begin{equation}\label{fmsetar}
\begin{cases}
&   \mathbf{y}_{t}+\sum_{J\in \{1,\cdots, L\}^{D}}y_{t,d}^{(J)}\left(a_{0}^{(J)}+\sum_{i=1}^{p_{J}}A_{i}^{(J)}\mathbf{y}_{t-i}+\Lambda^{(J)}\mathbf{f}_{t}\right)=\sum_{J\in \{1,\cdots, L\}^{D}}y_{t,d}^{(J)}\varepsilon_{t}^{(J)}\quad \\ 
\\
& \mathbf{f}_{t}=\sum_{\tau=1}^{q}\Xi_{\tau} \mathbf{f}_{t-\tau}+\eta_{t} \\ 
\end{cases}
\end{equation} is called a multivariate SETAR process with exogenous input denoted MSETARX $(L,p_{J},q;J\in \{1,\cdots, L\}^{D})$. The variables $(\mathbf{y}_{t})_{t}$ and $(\mathbf{f}_{t})_{t}$ in model \eqref{fmsetar} are endogenous and exogenous, respectively, and the econometrics significance of estimating the relationship between $(\mathbf{y}_{t})_{t}$ and $(\mathbf{f}_{t})_{t}$ is well known. The model equation \eqref{fmsetar} can be rewritten as 
\begin{equation}\label{fmsetar1}
\mathbf{y}_{t}+\sum_{J\in \{1,\cdots, L\}^{D}}y_{t,d}^{(J)}\left(a_{0}^{(J)}+\sum_{i=1}^{p_{J}}A_{i}^{(J)}\mathbf{y}_{t-i}+\Lambda^{(J)}\sum_{\tau=1}^{q}\Xi_{\tau} \mathbf{f}_{t-\tau}\right)=\sum_{J\in \{1,\cdots, L\}^{D}}y_{t,d}^{(J)}\omega_{t}^{(J)}
\end{equation} where $\omega_{t}^{(J)}=\varepsilon_{t}^{(J)}-\Lambda^{(J)}\eta_{t}$, $a_{0}^{(J)}$ and $\omega_{t}^{(J)}$ are $D\times 1$ vectors, $A_{i}^{(J)}$ are  $D\times D$ coefficient matrices, $\Lambda^{(J)}$ are $D\times \kappa$ coefficient matrices, $\Xi_{\tau}^{(J)}$ are $\kappa \times \kappa$ coefficient matrices, and $(f_{t})_{t}$ is $ \kappa \times 1$ vector. When $\Lambda^{(J)}=\mathbf{0}$ for all $J\in \{1,\cdots, L\}^{D}$, \eqref{fmsetar1} becomes a MSETAR model \eqref{msetar}. 

\noindent The representation in equation \eqref{fmsetar1} shows that the MSETARX $(L,p_{J},q;J\in \{1,\cdots, L\}^{D})$ model \eqref{fmsetar} has approximately the same structure as the MSETAR $(L,p_{J};J\in \{1,\cdots, L\}^{D})$ model \eqref{msetar} with exogenous variables or factors $(\mathbf{f}_t)_{t}$. For simplicity, we assume the exogenous inputs enter the model in a linear autoregressive fashion. It is worth pointing out that the dynamics of process $(\mathbf{f}_t)_{t}$ could be captured by suitable linear/nonlinear model, principal components, and among other model specifications. Unlike the multivariate threshold model in \cite{Tsy98}, the threshold space is of dimension equal to the dimension of the multivariate process. Thus there exists thresholds for all subprocess of the multivariate process \eqref{fmsetar1}.  In this case, the regime of the whole system is not necessarily determined by a single stationary subprocess, say $y_{it}$, as in \cite{Tsy98}. 

\begin{assumption}\label{A4}
Let $\{\varepsilon_{t}^{(J)}, \mathcal{F}_{t}\}$ and $\{\eta_{t}, \mathcal{F}_{t}\}; \forall J\in \{1,\cdots, L\}^{D}$ be two independent sequence of martingale difference with respect to an increasing sequence of 
$\sigma$-field $\{\mathcal{F}_{t}\}$ such that   
$$\sup_{t\geq 0} E [\lVert \varepsilon_{t+1}^{(J)} \rVert^{2}  | \mathcal{F}_{t}]=\breve{\Upsilon}_{\varepsilon} < \infty \quad a.s \quad and \quad 
\sup_{t\geq 0} E [\lVert \eta_{t+1} \rVert^{2}  | \mathcal{F}_{t}]=\breve{\Upsilon}_{\eta} < \infty \quad a.s $$ 
This ensures that $\{\omega_{t}^{(J)}, \mathcal{F}_{t}\}$ is a sequence of martingale difference with respect to an increasing sequence of 
$\sigma$-field $\{\mathcal{F}_{t}\}$ where
$$\sup_{t\geq 0} E [\lVert \omega_{t+1}^{(J)} \rVert^{2}  | \mathcal{F}_{t}]=\breve{\Upsilon}_{\omega} < \infty \quad a.s \quad .$$ 
\end{assumption} Simple orthogonality assumptions on the errors $\omega_{t}^{(J)}$ are insufficient to identify nonlinear models (\cite{Hansen04}) and as such it is 
important that Assumption \ref{A4} holds.

\noindent  Let $p=\max \{p_{J}| J\in \{1,\cdots, L\}^{D} \}$ and $q$ be the model orders for model\eqref{fmsetar1}. Now, suppose that $\omega_{t}^{(J)}$, and $p$ be regime independent.  We can rewrite model equation \eqref{fmsetar1} as 
\begin{equation}\label{fmsetar2}
 \mathbf{y}_{t}=\sum_{J\in \{1,\cdots, L\}^{D}}y_{t,d}^{(J)}\left( \breve{\Theta}^{(J)}\right)^{T}\breve{\Phi}_{t-1}+ \omega_{t}
\end{equation}  where  $\left( \breve{\Theta}^{(J)}\right)^{T}=-[a_{0}^{(J)}, A_{1}^{(J)},\cdots, A_{p}^{(J)}, \Lambda^{(J)}\Xi_{1},\Lambda^{(J)}\Xi_{2},\cdots,\Lambda^{(J)}\Xi_{q}]$, \\ $\breve{\Phi}_{t}^{T}=[1, \mathbf{y}_{t}^{T}, \mathbf{y}_{t-1}^{T},\cdots, \mathbf{y}_{t-p+1}^{T},\mathbf{f}_{t}^{T}, \mathbf{f}_{t-1}^{T},\cdots, \mathbf{f}_{t-q+1}^{T} ]$ and the notation $\zeta^{T}$ denotes the transpose of $\zeta$. 
We remark that the MSETARX model with the representation \eqref{fmsetar2} permits us to make use of the \cite{Arnold01} proposed adaptive parameter estimation algorithm for the MSETAR model \eqref{msetar}.

\section{On the Stationarity of MSETARX model}\label{sec:stationarity}
In this section, we establish the conditions for the existence of a solution for the model equation \eqref{fmsetar}. 
Let $p=\max \{p_{J}| J\in \{1,\cdots, L\}^{D} \}$ and $q$ be the model orders for model\eqref{fmsetar}. Now, suppose that $p$ and $q$ be regime independent and 
$a_{0}^{(J)}=0$ for each $J$. 
We rewrite model equation \eqref{fmsetar} in the form 
\begin{equation}\label{fmsetar3a}
\mathbf{y}_{t}=\sum_{J\in \{1,\cdots, L\}^{D}}y_{t,d}^{(J)} \Bigg( \Big( \breve{\Theta_{1}}^{(J)}\Big)^{T} \breve{\Phi}_{1, t-1}+ \Lambda^{(J)} \mathbf{f}_{t} \Bigg) + \sum_{J\in \{1,\cdots, L\}^{D}}y_{t,d}^{(J)}\varepsilon_{t}^{(J)}
\end{equation}   
\begin{equation}\label{fmsetar3b}
\mathbf{f}_{t}= \Big( \breve{\Theta_{2}}\Big)^{T} \breve{\Phi}_{2, t-1} + \eta_{t}
\end{equation} 
\noindent where  $\Big( \breve{\Theta_{1}}^{(J)}\Big)^{T}=-[A_{1}^{(J)},\cdots, A_{p}^{(J)}]$, 
$\Big( \breve{\Theta_{2}}\Big)^{T}=-[\Xi_{1},\Xi_{2},\cdots,\Xi_{q}]$,\\
$\breve{\Phi}_{1, t}^{T}=[\mathbf{y}_{t}^{T}, \mathbf{y}_{t-1}^{T},\cdots, \mathbf{y}_{t-p+1}^{T}]$ and 
$\breve{\Phi}_{2, t}^{T}=[\mathbf{f}_{t}^{T}, \mathbf{f}_{t-1}^{T},\cdots, \mathbf{f}_{t-q+1}^{T} ] .$

The equation model \eqref{fmsetar3a}-\eqref{fmsetar3b} can 
be represented as a nonlinear ARX model (\cite{EMDT97}) of the form :
\begin{equation}\label{fmsetar4}
\begin{cases}
&   \mathbf{y}_{t}= g_{1}^{(J)}(\mathbf{y}_{t-1}^{T},\cdots, \mathbf{y}_{t-p}^{T}) + g_{2}^{(J)}(\mathbf{f}_{t}^{T},\cdots, \mathbf{f}_{t-q}^{T}) + \sum_{J\in \{1,\cdots, L\}^{D}}y_{t,d}^{(J)}\varepsilon_{t}^{(J)} \quad \\ 
& \mathbf{f}_{t}= g_{3}(\mathbf{f}_{t-1}^{T},\cdots, \mathbf{f}_{t-q}^{T}) + \eta_{t}\\ 
\end{cases}
\end{equation} \noindent with  $g_{1}^{(J)}(\mathbf{y}_{t-1}^{T},\cdots, \mathbf{y}_{t-p}^{T})=\sum_{J\in \{1,\cdots, L\}^{D}}y_{t,d}^{(J)} \Big( \breve{\Theta_{1}}^{(J)}\Big)^{T} \breve{\Phi}_{1, t-1}, \quad$
$g_{2}^{(J)}(\mathbf{f}_{t}^{T},\cdots, \mathbf{f}_{t-q}^{T})=\sum_{J\in \{1,\cdots, L\}^{D}}y_{t,d}^{(J)} \Lambda^{(J)} \mathbf{f}_{t},\quad$ and 
$g_{3}(\mathbf{f}_{t-1}^{T},\cdots, \mathbf{f}_{t-q}^{T})=                                                                                                                                                                                                                                                                                                                                                                                                                    \Big( \breve{\Theta_{2}}\Big)^{T} \breve{\Phi}_{2, t-1}.$ 
The process $\{\mathbf{f}_{t},\mathbf{y}_{t}\}$ of the equation model \eqref{fmsetar4} is a Markov process. 

\begin{assumption}\label{A5}
We denote $\mathfrak{y}=(\mathbf{y}_{t-1}^{T},\cdots, \mathbf{y}_{t-p}^{T})$ and $\mathfrak{f}=(\mathbf{f}_{t}^{T},\cdots, \mathbf{f}_{t-q}^{T})$. 
The multivariate SETARX model \eqref{fmsetar4} satisfies the following: 

\begin{enumerate}
 \item The functions $g_{1}^{(J)}(\mathfrak{y})$, $g_{2}^{(J)}(\mathfrak{f})$, and $g_{3}(\mathfrak{f})$ for each $J\in \{1,\cdots, L\}^{D}$ are nonperiodic and bounded on compact sets, and 
 $g_{2}^{(J)}(\mathfrak{f})= O(\| \mathfrak{f} \|^{\gamma_{1}})$ as $\| \mathfrak{f} \| \rightarrow \infty $ for some real $\gamma_1$.
 \item Assumption \ref{A4} holds, the  $\sup_{t\geq 0} E [\| \eta_{t+1}\|^{\max (1,\gamma_{1}+\gamma_{2})}  | \mathcal{F}_{t}] < \infty $ for some $\gamma_{2} > 0.$
 \item There exist $\mathscr{A}^{(J)}=[\mathscr{A}_{1}^{(J)},\mathscr{A}_{2}^{(J)},\cdots, \mathscr{A}_{p}^{(J)}]$  and $\mathscr{B}=[\mathscr{B}_{1},\mathscr{B}_{2},\cdots, \mathscr{B}_{q-1}]$, 
 each of which may be the zero matrix, for each $J\in \{1,\cdots, L\}^{D}$, where $\mathscr{A}_{i}^{(J)}$ and $\mathscr{B}_{\tau}$ are matrices of dimension $D \times D$ and $\kappa \times \kappa$ respectively such that 
 $g_{1}^{(J)}(\mathfrak{y})=\mathfrak{y} \big(\mathscr{A}^{(J)}\big)^{T}+o(\|\mathfrak{y} \|)$ and  $g_{3}(\mathfrak{f})=\mathfrak{f} \big(\mathscr{B}\big)^{T}+o(\|\mathfrak{f} \|)$ 
 as $\|\mathfrak{y} \|$ and $\|\mathfrak{f} \| \rightarrow \infty.$ 
 Then the $Dp$-dimensional square matrix $\mathfrak{A}$ defined by $\mathbf{0}$ if $\mathscr{A}^{(J)}=0$ and by  
  $$ \mathfrak{A}=
 \begin{bmatrix}
  O_D & O_D & \cdots & O_D & \big(\mathscr{A}_{1}^{(J)}\big)^{T}\\
  I_D & O_D & \cdots & O_D & \big(\mathscr{A}_{2}^{(J)}\big)^{T} \\
  O_D & I_D & \cdots & O_D & \big(\mathscr{A}_{3}^{(J)}\big)^{T} \\
  \vdots & \vdots &\ddots&\vdots&\vdots\\
  O_D & O_D & \cdots & I_D & \big(\mathscr{A}_{p}^{(J)}\big)^{T} \\
  \end{bmatrix}
$$  otherwise, and the $\kappa q$-dimensional square matrix $\mathfrak{B}$ be defined by   
$$ \mathfrak{B}=
 \begin{bmatrix}
  O_{\kappa} & O_{\kappa} & \cdots & O_{\kappa} & \big(\mathscr{B}_{1}\big)^{T}\\
  I_{\kappa} & O_{\kappa} & \cdots & O_{\kappa} & \big(\mathscr{B}_{2}\big)^{T} \\
  O_{\kappa} & I_{\kappa} & \cdots & O_{\kappa} & \big(\mathscr{B}_{3}\big)^{T} \\
  \vdots & \vdots &\ddots&\vdots&\vdots\\
  O_{\kappa} & O_{\kappa} & \cdots & I_{\kappa} & \big(\mathscr{B}_{q}                                                                                                                                                                                                                                                                                                                                                                                                                                                                                                                                                                                                                                                                                                                                                                                                                                                                                                                                                                                                          \big)^{T} \\
  \end{bmatrix}
$$ satisfy $\varrho(\mathfrak{A}) < 1$ and $\varrho (\mathfrak{B}) < 1$, where $\varrho$ denotes the spectral radius, $O_{\iota}$ denotes the  $\iota$-dimensional zero square matrix and 
 $I_{\iota}$ denotes the  $\iota$-dimensional unit square matrix.
 
\end{enumerate}
\end{assumption}

\begin{lemma}\label{L1}
 Under Assumption \ref{A5}, $\{\mathbf{f}_{t}, \mathbf{y}_{t}\}$ of the multivariate SETARX model \eqref{fmsetar3a}-\eqref{fmsetar3b} represented as a nonlinear ARX 
 model \eqref{fmsetar4} is $\alpha$-mixing with mixing coefficient $\alpha (k) \sim e^{-\beta k}$ for some $\beta > 0.$
\end{lemma}

\begin{proof}
 The result is known as in Lemma 3.1 in \cite{EMDT97} and thus we do not provide the proof since it is roughly same. We refer the interested reader to remarks after Assumption 3.3 and Lemma 3.1 in \cite{EMDT97} and the references therein. 
\end{proof}

\begin{remark}
 Lemma \ref{L1} provides sufficient conditions for the multivariate SETARX process \eqref{fmsetar4} to be stationary (\cite{EMDT97,DTj90,Pham86}). The proof of this Lemma as in 
 Lemma 3.1 in \cite{EMDT97} implies geometric ergodicity and stronger conclusion of absolute regularity with an exponentially decreasing rate (\cite{DTj90,Pham86,Twee75,Twee88}).
\end{remark}

\begin{lemma}\label{L2}
Let $a_{0}^{(J)}=0$ for each $J$ in model \eqref{msetar} and $p=1$. Assume that there is a $D$-cycle of indexes $j_{1}\rightarrow j_{2} \rightarrow j_{3} \rightarrow \cdots \rightarrow j_{D}\rightarrow j_{1}$ with the notation $A_{1}^{(j_s)}$ corresponding to $A_{1}^{(j_s)}$(mod the $D$-cycle) so that $A_{1}^{(j_{s+1})}=A_{1}^{(j_1)}$. The process $\{\mathbf{y}_{t}\}$ of the  multivariate SETAR model \eqref{msetar} is 
geometrically ergodic if $$\varrho \Big(\prod_{s=1}^{D} -A_{1}^{(j_s)} \Big) < 1$$ where $\varrho$ denotes the spectral radius and the product notation $\prod_{s=m}^{n} A^{(j_s)}=A^{(j_m)} \cdots A^{(j_{m+1})}A^{(j_m)}$ is 
interpreted as the identity matrix if $n=m-1.$
\end{lemma}

\begin{proof}
 The result about geometric ergodicity follows from Theorem 4.5 and equation model (4.12) in \cite{DTj90} with $A_{i_s}=-A_{1}^{(j_s)}$ and $k=D.$
\end{proof}

\section{Estimation of model parameters}\label{sec:estimation}

In this section, we assume that assumption \ref{A5} and Lemma \ref{L2} are satisfied. We also assume the model orders $p$, $q$, $d$, and $L$, of model \eqref{fmsetar}-\eqref{fmsetar1}-\eqref{fmsetar2} are known. 
Let model \eqref{fmsetar} be represented as a MSETAR $(L,p_{J};J\in \{1,\cdots, L\}^{D})$ model \eqref{msetar} with exogenous variables or factors as in model \eqref{fmsetar1}-\eqref{fmsetar2}.
We propose to use estimation procedures based on the standard LSE approach and the concept of self-tuning regulators used in the study of adaptive control of stochastic linear systems (see \cite{Kumar86}). \cite{Arnold01} has shown that algorithms for estimation of parameters based on the stochastic gradient principles for linear systems are also suitable for nonlinear systems. 
Alternatively, following \cite{ZCEM00}, one can use local linear fitting plus the projection method to estimate components $g_{1}^{(J)}(\cdot)$ and $g_{2}^{(J)}(\cdot)$ of model equation \eqref{fmsetar4}. The function  
$g_{3}(\cdot)$ can then be estimated directly using a standard approach or by kernel-type estimation (\cite{EMDT95}). 

\subsection{Standard LSE Algorithm for Parameter Estimation }\label{algorithm2}
Consider the MSETARX $(L,p_{J},q;J\in \{1,\cdots, L\}^{D})$ model in equation \eqref{fmsetar2}: 
\begin{equation}\label{fmsetar2**}
 \mathbf{y}_{t}=\sum_{J\in \{1,\cdots, L\}^{D}}y_{t,d}^{(J)}\left( \breve{\Theta}^{(J)}\right)^{T}\breve{\Phi}_{t-1}+ \omega_{t}
\end{equation}  where  $\left( \breve{\Theta}^{(J)}\right)^{T}=-[a_{0}^{(J)}, A_{1}^{(J)},\cdots, A_{p}^{(J)}, \Lambda^{(J)}\Xi_{1},\Lambda^{(J)}\Xi_{2},\cdots,\Lambda^{(J)}\Xi_{q}]$, \\ $\breve{\Phi}_{t}^{T}=[1, \mathbf{y}_{t}^{T}, \mathbf{y}_{t-1}^{T},\cdots, \mathbf{y}_{t-p+1}^{T},\mathbf{f}_{t}^{T}, \mathbf{f}_{t-1}^{T},\cdots, \mathbf{f}_{t-q+1}^{T} ]$
with the autoregressive orders $p_{J}, q$, delay $d$, and thresholds known. Then the LSE is $ \hat{\breve{\Theta}}^{(J)}=\sum_{J\in \{1,\cdots, L\}^{D}}y_{t,d}^{(J)}\left(\breve{\Phi}_{t-1}^{T}\breve{\Phi}_{t-1}\right)^{-1}\breve{\Phi}_{t-1}^{T}\mathbf{y}_{t}$. Following \cite{Kumar86} presentation of the stochastic gradient algorithm for ARX systems, the true parameter $\breve{\Theta}^{(J)}$ can also be estimated by the LSE using the recursion, 
\begin{equation}\label{rlse0} 
\hat{\breve{\Theta}}_{k+1}^{(J)}= \hat{\breve{\Theta}}_{k}^{(J)}+y_{k+1,d}^{(J)}R_{k}^{-1}\breve{\Phi}_{k}\big( \mathbf{y}_{k+1}^{T} - \breve{\Phi}_{k}^{T} \hat{\breve{\Theta}}_{k}^{(J)} \big)
 \end{equation}
\begin{equation}\label{rlse1} 
R_{k}=\sum_{J\in \{1,\cdots, L\}^{D}}\sum_{i=0}^{k}y_{k+1,d}^{(J)} \breve{\Phi}_{i}\breve{\Phi}_{i}^{T}
\end{equation}

\subsection{Algorithm for Adaptive Parameter Estimation}\label{algorithm}
Let $ 0< \alpha \leq 1$, $0 < \upsilon^{(J)} \leq 1$, $p^{*} = \max \{p,d,q\}$, and $\breve{\Theta}$ be the coefficients of the MSETARX $(L,p_{J},q;J\in \{1,\cdots, L\}^{D})$ model in equation \eqref{fmsetar2}. 
\begin{equation*}
\breve{\Theta}_{k}^{(J)}=0; \quad k\leq p^{*}
\end{equation*}
\begin{equation*}
\breve{\Theta}_{k+1}^{(J)}=\breve{\Theta}_{k}^{(J)} + y_{k+1,d}^{(J)}\frac{\alpha \breve{\Phi}_{k}}{s_{k}^{(J)}}\big( \mathbf{y}_{k+1}^{T} - \breve{\Phi}_{k}^{T} \breve{\Theta}_{k}^{(J)} \big) ; \quad k\geq p^{*}
\end{equation*}
\begin{equation}\label{algo1} 
r_{k}^{(J)}=
\begin{cases}
&   1; \quad k < p^{*}\\ 
& r_{k-1}^{(J)}+ \sum_{J\in \{1,\cdots, L\}^{D}}y_{k+1,d}^{(J)} \|\breve{\Phi}_{k}\|^{2}; \quad k\geq p^{*} . \\
\end{cases}\\
 \end{equation}
\begin{equation}\label{algo2} 
s_{k}^{(J)}=
\begin{cases}
&    1; \quad k < p^{*}\\ 
&  s_{k-1}^{(J)}+ \mathbf{y}_{k+1}^{(J)} \Big(\max \{\upsilon^{(J)} r_{k-1}^{(J)},1 \} + \|\mathbf{y}_{k}\|^{2} - s_{k-1}^{(J)} \Big); \quad k\geq p^{*} . \\
\end{cases}\\
 \end{equation}
This algorithm \ref{algorithm} corresponds to the adaptive parameter estimation algorithm proposed by \cite{Arnold01}, with the control sequence being $(s_{k}^{(J)})^{-1}$ instead of $(r_{k}^{(J)})^{-1}$. The simulation results presented by the authors showed that as the control sequence $(r_{k}^{(J)})^{-1}$ becomes large, a further progress towards the true coefficients is prevented or slowed down 
since this control sequence which weight the prediction error decrease too fast. The \textit{relaxed control sequence} $(s_{k}^{(J)})^{-1}$ have similar properties as $(r_{k}^{(J)})^{-1}$ with the convergence spend decreased by the factors $\upsilon^{(J)}$ and in particular, improves the estimation accuracy (\cite{Arnold01}). This algorithm was applied in \cite{Lutz06} for the analysis of biomedical signals.

\subsection{Simulations}\label{sec:simulations}
In this section, we carry out a simulation exercise to study the performance of the parameter estimation algorithm presented in Section \ref{algorithm2}\&\ref{algorithm} on MSETARX models. In this respect, we consider two data generating process (DGP) according to the following: 

\begin{enumerate}
\item Consider a simulated $50,000$ points of a two-dimensional MSETARX process with six-regimes, delay $d=6$, $\Lambda^{(J)}=\mathbf{0}$ for all $J\in \{1,\cdots, L\}^{D}$ in equation \eqref{fmsetar1}, standard normal noise $N(0,1)$ added to all regimes and autoregressive order $p=3$ defined by: 
\begin{equation}
 \mathbf{y}_{t}=
\begin{cases}
&a_{0}^{(1)}+ \mathscr{A}^{(1)}_{1}\mathbf{y}_{t-1}+\mathscr{A}^{(1)}_{2}\mathbf{y}_{t-2}+\mathscr{A}^{(1)}_{3}\mathbf{y}_{t-3}+ \omega_{t}; \quad R_{1}^{i}:= [-\infty,-0.50)\times [-\infty,0)\\ 
  
&a_{0}^{(2)}+ \mathscr{A}^{(2)}_{1}\mathbf{y}_{t-1}+\mathscr{A}^{(2)}_{2}\mathbf{y}_{t-2}+\mathscr{A}^{(2)}_{3}\mathbf{y}_{t-3}+ \omega_{t}; \quad R_{2}^{i}:= [-\infty,-0.50)\times [0, \infty)\\
 
&a_{0}^{(3)}+ \mathscr{A}^{(3)}_{1}\mathbf{y}_{t-1}+\mathscr{A}^{(3)}_{2}\mathbf{y}_{t-2}+\mathscr{A}^{(3)}_{3}\mathbf{y}_{t-3}+ \omega_{t};  \quad R_{3}^{i}:= [-0.50,0.50)\times [-\infty,0)\\   

&a_{0}^{(4)}+ \mathscr{A}^{(4)}_{1}\mathbf{y}_{t-1}+\mathscr{A}^{(4)}_{2}\mathbf{y}_{t-2}+\mathscr{A}^{(4)}_{3}\mathbf{y}_{t-3}+ \omega_{t};  \quad R_{4}^{i}:= [-0.50,0.50)\times (0.00,\infty)\\ 

&a_{0}^{(5)}+ \mathscr{A}^{(5)}_{1}\mathbf{y}_{t-1}+\mathscr{A}^{(5)}_{2}\mathbf{y}_{t-2}+\mathscr{A}^{(5)}_{3}\mathbf{y}_{t-3}+ \omega_{t};  \quad R_{5}^{i}:= [0.50,\infty)\times [-\infty,0.00)\\ 

&a_{0}^{(6)}+ \mathscr{A}^{(6)}_{1}\mathbf{y}_{t-1}+\mathscr{A}^{(6)}_{2}\mathbf{y}_{t-2}+\mathscr{A}^{(6)}_{3}\mathbf{y}_{t-3}+ \omega_{t};  \quad R_{6}^{i}:= [0.50, \infty)\times [0.00,\infty)\\ 
\end{cases}
\end{equation}
  \begin{description}
   \item[Regime 1] $R_{1}^{i}:= [-\infty,-0.50)\times [-\infty,0)$, $\mathscr{A}^{(1)}_{1}=\begin{pmatrix}
  -0.02 & 0.00\\
  0.00 & 0.30\\
  \end{pmatrix}$, $\mathscr{A}^{(1)}_{2}=\begin{pmatrix}
  0.53 & 0.00\\
  0.00 & 0.30\\
  \end{pmatrix}$, $\mathscr{A}^{(1)}_{3}=\begin{pmatrix}
  0.00 & 0.53\\
  0.00 & 0.30\\
  \end{pmatrix}$, $a_{0}^{(1)}=\begin{pmatrix}
  0.74\\
  -0.20\\
  \end{pmatrix}$
   \item[Regime 2] $R_{2}^{i}:= [-\infty,-0.50)\times [0, \infty)$, $\mathscr{A}^{(2)}_{1}=\begin{pmatrix}
  -0.02 & 0.00\\
  0.00 & 0.30\\
  \end{pmatrix}$, $\mathscr{A}^{(2)}_{2}=\begin{pmatrix}
  0.53 & 0.00\\
  0.00 & 0.30\\
  \end{pmatrix}$, $\mathscr{A}^{(2)}_{3}=\begin{pmatrix}
  0.00 & 0.53\\
  0.00 & 0.30\\
  \end{pmatrix}$, $a_{0}^{(2)}=\begin{pmatrix}
  -0.75\\
  -0.20\\
  \end{pmatrix}$ 
   \item[Regime 3] $R_{3}^{i}:= [-0.50,0.50)\times [-\infty,0)$, $\mathscr{A}^{(3)}_{1}=\begin{pmatrix}
  -0.94 & 0.00\\
  0.00 & 0.30\\
  \end{pmatrix}$, $\mathscr{A}^{(3)}_{2}=\begin{pmatrix}
  0.85 & 0.00\\
  0.00 & 0.30\\
  \end{pmatrix}$, $\mathscr{A}^{(3)}_{3}=\begin{pmatrix}
  0.00 & 0.85\\
  0.00 & 0.30\\
  \end{pmatrix}$, $a_{0}^{(3)}=\begin{pmatrix}
  1.15\\
  -0.20\\
  \end{pmatrix}$ 
   \item[Regime 4] $R_{4}^{i}:= [-0.50,0.50)\times (0.00,\infty)$, $\mathscr{A}^{(4)}_{1}=\begin{pmatrix}
  -0.94 & 0.00\\
  0.00 & 0.30\\
  \end{pmatrix}$, $\mathscr{A}^{(4)}_{2}=\begin{pmatrix}
  0.85 & 0.00\\
  0.00 & 0.30\\
  \end{pmatrix}$, $\mathscr{A}^{(4)}_{3}=\begin{pmatrix}
  0.00 & 0.85\\
  0.00 & 0.30\\
  \end{pmatrix}$, $a_{0}^{(4)}=\begin{pmatrix}
  0.74\\
  0.20\\
  \end{pmatrix}$ 
   \item[Regime 5] $R_{5}^{i}:= [0.50,\infty)\times [-\infty,0.00)$, $\mathscr{A}^{(5)}_{1}=\begin{pmatrix}
  -1.10 & 0.00\\
  0.00 & 0.30\\
  \end{pmatrix}$, $\mathscr{A}^{(5)}_{2}=\begin{pmatrix}
  -0.30 & 0.00\\
  0.00 & 0.30\\
  \end{pmatrix}$, $\mathscr{A}^{(5)}_{3}=\begin{pmatrix}
  0.00 & -0.30\\
  0.00 & 0.30\\
  \end{pmatrix}$, $a_{0}^{(5)}=\begin{pmatrix}
  -0.75\\
  0.20\\
  \end{pmatrix}$
   \item[Regime 6] $R_{6}^{i}:= [0.50, \infty)\times [0.00,\infty)$, $\mathscr{A}^{(6)}_{1}=\begin{pmatrix}
  -1.10 & 0.00\\
  0.00 & 0.30\\
  \end{pmatrix}$, $\mathscr{A}^{(6)}_{2}=\begin{pmatrix}
  0.30 & 0.00\\
  0.00 & 0.30\\
  \end{pmatrix}$, $\mathscr{A}^{(6)}_{3}=\begin{pmatrix}
  0.00 & 0.30\\
  0.00 & 0.30\\
  \end{pmatrix}$, $a_{0}^{(6)}=\begin{pmatrix}
  1.15\\
  0.20\\
  \end{pmatrix}$, 
  \end{description} and a signal section is shown in Figure \ref{fig:2dimMSETARprocess}. In \ref{estimResults}, autoregressive coefficient estimates obtained via the LSE algorithm is provided.
  
 \item Consider a three--regime ($L=3$) bivariate ($D=2$) MSETARX $(L,p_{J},q;J\in \{1,\cdots, L\}^{D})$ model with a bivariate exogenous input ($\kappa=2$), model orders be unit ($p=\max \{p_{J}| J\in \{1,\cdots, L\}^{D} \}=1$, $q=1$) and delay $d=1$: 
\begin{equation}
 \mathbf{y}_{t}=
\begin{cases}
& \mathscr{A}^{(1)}_{1}\mathbf{y}_{t-1}+ \Lambda^{(1)}\Xi_{1} \mathbf{f}_{t-1}+ \omega_{t}; \quad (y_{t-1})_{2}\leq -0.5 \quad \\ 
  
& \mathscr{A}^{(2)}_{1}\mathbf{y}_{t-1}+ \Lambda^{(2)}\Xi_{1} \mathbf{f}_{t-1}+ \omega_{t}; \quad (y_{t-1})_{2}\in (-0.5,0.5]\\
 
& \mathscr{A}^{(3)}_{1}\mathbf{y}_{t-1}+ \Lambda^{(3)}\Xi_{1} \mathbf{f}_{t-1}+ \omega_{t};  \quad (y_{t-1})_{2}\geq 0.5\\   
\end{cases}
\end{equation} where $\mathscr{A}^{(1)}_{1}=\begin{pmatrix}
  -0.3 & 0.6\\
  -0.7 & 0.4\\
  \end{pmatrix}$, $\mathscr{A}^{(2)}_{1}=\begin{pmatrix}
  1.5 & -1\\
  0.2 & 0.3\\
  \end{pmatrix}$, $\mathscr{A}^{(3)}_{1}=\begin{pmatrix}
  0.3 & -0.1\\
  0.2 & 0.6\\
  \end{pmatrix}$, $\Xi_{1}=\begin{pmatrix}
  0.5 & 0\\
  0.3 & 0\\
  \end{pmatrix}$, $\Lambda^{(1)}=\begin{pmatrix}
  0.2 & 0\\
  0 & 0\\
  \end{pmatrix}$, $\Lambda^{(2)}=\begin{pmatrix}
  0.3 & 0\\
  0 & 0.2\\
  \end{pmatrix}$, $\Lambda^{(3)}=\begin{pmatrix}
  0.8 & 0\\
  0 & 0\\
  \end{pmatrix}$. It is worth noting that the multivariate process $\mathbf{y}_{t}$ is unstable in the inner regime and only the second subprocess 
  determines the current regime.   
\end{enumerate}

\begin{figure}[!ht]
  \centering
\includegraphics[width=1.2\textwidth]{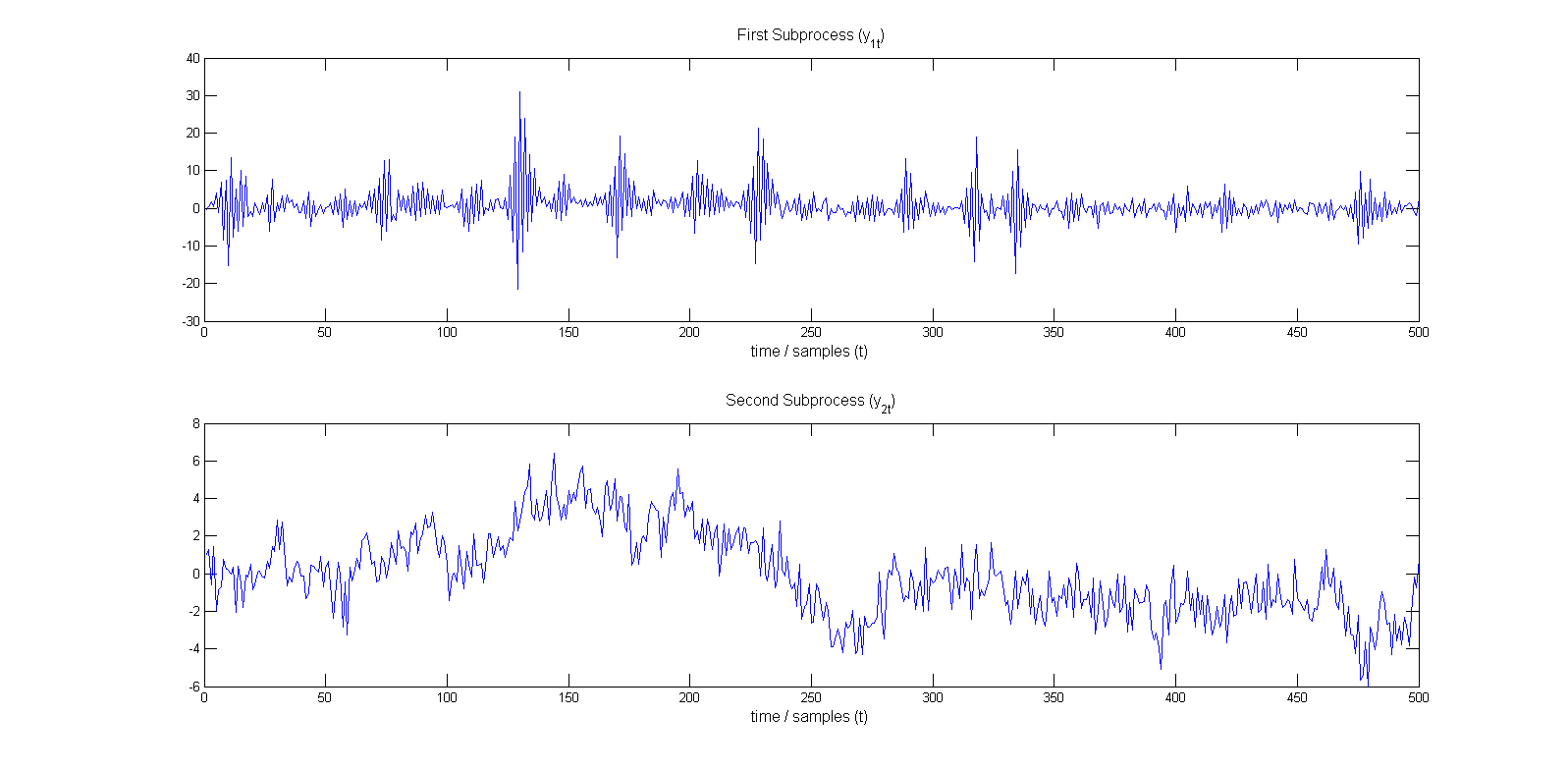}
 \caption{Two-dimensional MSETARX process with six-regimes, delay $d=6$, autoregressive order $p=3$ and $\Lambda^{(J)}=\mathbf{0}$ for all $J\in \{1,\cdots, L\}^{D}$ in equation \eqref{fmsetar1}. This is a signal section of $500$ time samples of the multivariate process.}
  \label{fig:2dimMSETARprocess}
\end{figure}

\section{Concluding remarks}\label{sec:conclusion}
The recent financial crisis of 2007-2009 has lead to a need for regulators and policy makers to understand and track systemic linkages. As the events following the turmoil in financial markets unfolded, it became evident that modern financial systems exhibit a high degree of interdependence and nonlinearity making it difficult in predicting the consequences of such an intertwined system. In this study, we define a nonlinear multivariate SETARX model useful in modeling economic relationships and to capture non-linear dynamics such as regime switching and asymmetric responses to shocks. We then present an estimation procedure for the parameters. 

In general, testing linearity is the first step of a proper modelling strategy of nonlinear models as it is possible that a linear model could adequately capture the relationship considered. Nonlinear models are usually not identified when the underlying process is linear (\cite{Terasvirta94, Terasvirta13, Hansen99, Tsy98, pma5}). The proposed test statistic for detecting threshold nonlinearity in vectors time series and the procedure for building multivariate threshold models discussed in \cite{Tsy98} could be performed on each subprocess in the MSETARX model setting. 
In this case, \cite{Arnold01} suggests a reasonable choice of the delay to be $d^{*}=argmax \{ \sum_{i}^{D} \mathscr{C}^{(i)}(d)~ |~ d \in \{1,\cdots,d_{max}\} \}$ where $\mathscr{C}^{(i)}(d)$ is the value of the test statistic (\cite{Tsy98}) for each subprocess $i$. One could apply the Wald test procedure used in \cite{Balke00}, which is a generalisation of \cite{Hansen96} approach, to test linearity. Another possibility of testing linear VAR model against a MSETARX model would be to generalise the approach the approach of \cite{Strikholm06} to multivariate models.
 
After the parameter estimation of model \eqref{fmsetar1}, it is necessary to evaluate the model by appropriate misspecification tests before 
putting it into practice. The general purpose is to find out if the assumptions made in the estimation step appear satisfied (\cite{Tsy98, Strikholm06, Hansen97, Hansen00}). For more details about modelling strategies and issues of vector threshold autoregressive models, we refer interested readers 
to \cite{Tsy98, Hansen11}. This model could be very useful in studying huge data sets such as the analysis of high-frequency financial data. 
 
Many problems remain open for the multivariate SETARX models. For example, establishing a testing procedure in determining the number of regimes and the specification of the threshold space will required a careful investigation.

\section*{Acknowledgement}
This research is supported by the Erasmus Mundus Fellowship. We are grateful to Lutz Leistritz for his support. 

\appendix

\section{Estimation Results}\label{estimResults}
We provide below the estimation of parameters obtained via LSE algorithm in Section \ref{algorithm2} on the first simulated process in Section \ref{sec:simulations}. The regime time corresponds to the number of temporal samples, where the multivariate process stayed in each regime. 
 \begin{description}
   \item[Regime 1] $R_{1}^{i}:= [-\infty,-0.50)\times [-\infty,0)$, $\hat{\mathscr{A}}^{(1)}_{1}=\begin{pmatrix}
  -0.0278 & -0.0169\\
  0.0027 & 0.2812\\
  \end{pmatrix}$, $\hat{\mathscr{A}}^{(1)}_{2}=\begin{pmatrix}
  0.5275 & 0.0025\\
  0.0013 & 0.3073\\
  \end{pmatrix}$, $\hat{\mathscr{A}}^{(1)}_{3}=\begin{pmatrix}
  0.0069 & 0.5419\\
  -0.0005 & 0.3046\\
  \end{pmatrix}$, $\hat{a}_{0}^{(1)}=\begin{pmatrix}
  0.7399\\
  -0.2012\\
  \end{pmatrix}$, (regime time: 10927).
   \item[Regime 2] $R_{2}^{i}:= [-\infty,-0.50)\times [0, \infty)$, $\hat{\mathscr{A}}^{(2)}_{1}=\begin{pmatrix}
  -0.0156 & 0.0009\\
  0.0033 & 0.2935\\
  \end{pmatrix}$, $\hat{\mathscr{A}}^{(2)}_{2}=\begin{pmatrix}
  0.5317 & -0.0051\\
  0.0043 & 0.3102\\
  \end{pmatrix}$, $\hat{\mathscr{A}}^{(2)}_{3}=\begin{pmatrix}
  -0.0012 & 0.5173\\
  -0.0021 & 0.2904\\
  \end{pmatrix}$, $\hat{a}_{0}^{(2)}=\begin{pmatrix}
  -0.7404\\
  -0.1951\\
  \end{pmatrix}$, (regime time: 8770). 
   \item[Regime 3] $R_{3}^{i}:= [-0.50,0.50)\times [-\infty,0)$, $\hat{\mathscr{A}}^{(3)}_{1}=\begin{pmatrix}
  -0.9417 & -0.0008\\
  0.0143 & 0.2859\\
  \end{pmatrix}$, $\hat{\mathscr{A}}^{(3)}_{2}=\begin{pmatrix}
  0.8602 & 0.0040\\
  0.0211 & 0.3003\\
  \end{pmatrix}$, $\hat{\mathscr{A}}^{(3)}_{3}=\begin{pmatrix}
  0.0067 & 0.8483\\
  -0.0004 & 0.3014\\
  \end{pmatrix}$, $\hat{a}_{0}^{(3)}=\begin{pmatrix}
  1.1337\\
  -0.2408\\
  \end{pmatrix}$, (regime time: 3932). 
   \item[Regime 4] $R_{4}^{i}:= [-0.50,0.50)\times (0.00,\infty)$, $\hat{\mathscr{A}}^{(4)}_{1}=\begin{pmatrix}
  -0.9302 & -0.0210\\
  -0.0023 & 0.3142\\
  \end{pmatrix}$, $\hat{\mathscr{A}}^{(4)}_{2}=\begin{pmatrix}
  0.8631 & 0.0033\\
  -0.0135 & 0.3116\\
  \end{pmatrix}$, $\hat{\mathscr{A}}^{(4)}_{3}=\begin{pmatrix}
  0.0066 & 0.8497\\
  0.0079 & 0.2789\\
  \end{pmatrix}$, $\hat{a}_{0}^{(4)}=\begin{pmatrix}
  0.7101\\
  0.1960\\
  \end{pmatrix}$, (regime time: 3235). 
   \item[Regime 5] $R_{5}^{i}:= [0.50,\infty)\times [-\infty,0.00)$, $\hat{\mathscr{A}}^{(5)}_{1}=\begin{pmatrix}
  -1.1008 & -0.0056\\
  0.0032 & 0.2923\\
  \end{pmatrix}$, $\hat{\mathscr{A}}^{(5)}_{2}=\begin{pmatrix}
  -0.2918 & -0.0012\\
  0.0019 & 0.3106\\
  \end{pmatrix}$, $\hat{\mathscr{A}}^{(5)}_{3}=\begin{pmatrix}
  0.0066 & -0.2971\\
  0.0029 & 0.2958\\
  \end{pmatrix}$, $\hat{a}_{0}^{(5)}=\begin{pmatrix}
  -0.7595\\
  0.1927\\
  \end{pmatrix}$, (regime time: 9697).
   \item[Regime 6] $R_{6}^{i}:= [0.50, \infty)\times [0.00,\infty)$, $\hat{\mathscr{A}}^{(6)}_{1}=\begin{pmatrix}
  -1.0995 & 0.0087\\
  0.0013 & 0.3147\\
  \end{pmatrix}$, $\hat{\mathscr{A}}^{(6)}_{2}=\begin{pmatrix}
  0.3011 & -0.0150\\
  0.0029 & 0.2904\\
  \end{pmatrix}$, $\hat{\mathscr{A}}^{(6)}_{3}=\begin{pmatrix}
  -0.0004 & 0.3033\\
  0.0026 & 0.2996\\
  \end{pmatrix}$, $\hat{a}_{0}^{(6)}=\begin{pmatrix}
  1.1508\\
  0.1942\\
  \end{pmatrix}$, (regime time: 13433).
  \end{description}

%
%
%
%
%


\bibliographystyle{elsarticle-harv}
\bibliography{reference1.bib}







\end{document}